\documentclass[a4paper,english,cleveref, autoref]{lipics-v2019}

\hideLIPIcs
\nolinenumbers

\title{Decidability of Intersection Emptiness versus Regular
 Separability} 

\author{Ramanathan S. Thinniyam}
{Max Planck Institute for Software Systems (MPI-SWS), Germany} 
{thinniyam@mpi-sws.org} 
{https://orcid.org/0000-0002-9926-0931} 
{} 

\author{Georg Zetzsche}
{Max Planck Institute for Software Systems (MPI-SWS), Germany} 
{georg@mpi-sws.org} 
{https://orcid.org/0000-0002-6421-4388} 
{} 

\authorrunning{R. S. Thinniyam and G. Zetzsche}

\Copyright{Ramanathan S. Thinniyam and Georg Zetzsche}



\relatedversion{}

\supplement{}

\acknowledgements{We thank Lorenzo Clemente and Wojciech Czerwi\'{n}ski for fruitful discussions.}

\EventEditors{John Q. Open and Joan R. Access}
\EventNoEds{2}
\EventLongTitle{39th IARCS Annual Conference on
Foundations of Software Technology and Theoretical Computer Science 
(FSTTCS 2019)}
\EventShortTitle{FSTTCS 2019}
\EventAcronym{FSTTCS}
\EventYear{2019}
\EventDate{December 11--13, 2019}
\EventLocation{IIT Bombay, India}
\EventLogo{}
\SeriesVolume{42}
\ArticleNo{23}
\usepackage{amsfonts, amsmath, amssymb, amsthm }
\usepackage{color}
\usepackage{mathrsfs}

\usepackage{appendix}
\usepackage{graphicx}
\usepackage{tikz}
\usepackage{mathtools}
\usetikzlibrary{shapes,positioning,fit,calc,decorations.text,decorations.pathmorphing,backgrounds}
\usepackage{stmaryrd}
\usepackage{xparse}
\usepackage{xspace}
\usepackage{todonotes}
\usepackage{tikz}
\usetikzlibrary{automata}

\theoremstyle{definition}
\newtheorem{thm}{Theorem}[section]

\newtheorem{defn}[thm]{Definition}
\newtheorem{lem}[thm]{Lemma}

\newtheorem{prob}[thm]{Problem}
\newtheorem{rem}[thm]{Remark}

\title{Regular Separability and Intersection Emptiness are Independent
Problems}
\date{\today}
\newcommand{\calM}{\mathcal{M}}	
\newcommand{\calT}{\mathcal{T}}	
\newcommand{\calV}{\mathcal{V}}	
\newcommand{\calR}{\mathcal{R}}	
\newcommand{\calA}{\mathcal{A}}	
\newcommand{\calH}{\mathcal{H}}	
\newcommand{\calD}{\mathcal{D}}	
\newcommand{\calC}{\mathcal{C}} 
\newcommand{\calP}{\mathcal{P}} 

\newcommand{\N}{\mathbb{N}}
\newcommand{\Z}{\mathbb{Z}}

\newcommand{\rev}[1]{{#1}^{\mathsf{rev}}}

\newcommand{\regsep}[2]{#1\mathord{\mid}#2}

\newcommand{\pseudo}[1]{\mathsf{pseudo}{#1}}
\newcommand{\power}[1]{\mathsf{power}{#1}}

\newcommand{\cntlang}[1]{\mathcal{I}(#1)}

\newcommand{\Inter}[2]{\mathsf{IE}(#1,#2)}
\newcommand{\RegSep}[2]{\mathsf{RS}(#1,#2)}

\newcommand{\Empty}[1]{\mathsf{Empty}(#1)}

\newcommand{\Inf}[1]{\mathsf{Inf}(#1)}

\newcommand{\ltr}[1]{\mathtt{#1}}

\newcommand{\langof}[1]{\mathsf{L}(#1)}

\newcommand{\reset}{\mathsf{r}}

\newcommand{\stacks}[2]{\mathcal{S}^{#1}_{#2}}
\newcommand{\instr}[2]{I^{#1}_{#2}}
\newcommand{\push}[1]{\mathsf{push}_{#1}}
\newcommand{\pop}[1]{\mathsf{pop}_{#1}}
\newcommand{\tops}{\mathsf{top}}
\newcommand{\rew}[1]{\mathsf{rew}_{#1}}

\DeclareDocumentCommand{\autstep}{O{}}{%
  \rightarrow_{#1}
}
\DeclareDocumentCommand{\autsteps}{O{}}{%
  \rightarrow^*_{#1}
}
\DeclareDocumentCommand{\autstepp}{O{} m}{%
  \xrightarrow{#2}_{#1}
}
\DeclareDocumentCommand{\autstepps}{m}{%
  \xrightarrow[*]{#1}
}

\newcommand{\lelex}{<_{\mathsf{lex}}}
\tikzset{gadget/.style={->,>=stealth,initial text=,minimum size=0pt,auto,on grid,scale=1,inner sep=1pt,node distance=3cm}}
\tikzset{every state/.style={minimum size=0pt,inner sep=1pt,fill=black!10,draw=black!70,thick}}

\newcommand{\newhopa}{
\begin{tikzpicture}[gadget, node distance=2.1cm]
  \node[state] (init') {$q'_0$};
  \node[state] (init) [right=5.5cm of init'] {$q_0$};
  \node[state] (final) [right=1.5cm of init] {$q_f$};
  \node[state] (p) [right= 2.5cm of final] {$p$};
  \node[state] (final') [accepting by arrow, right= 2.5cm of p] {$q_f'$};

  \draw[dashed,rounded corners=6pt,color=black!70,thick] ($(init.west)+(-0.2,-0.4)$) rectangle ( $(final.east)+(0.2,0.4)$ ); 

  \path (init') edge node {$\ltr{1}|\push{k+2}\rew{\#}\push{k+1}\rew{\bot}$} (init);
  \path (final) edge node {$\varepsilon|\pop{k+1}$} (p);
  \path (p) edge [loop above] node {\begin{tabular}{c}$\varepsilon|\ltr{0}|\pop{k+1}\push{k+2}$ \\ $\ltr{0}|\ltr{1}|\pop{k+1}\push{k+2}$\end{tabular}} (p);
  \path (p) edge [loop below] node {$\varepsilon|\#|\pop{k+2}$} (p);

  \path (p) edge node {$\varepsilon|\bot'|\varepsilon$} (final');
\end{tikzpicture}
}

\newcommand{\newedge}{
  \begin{tikzpicture}[gadget, node distance=2.1cm]
    \node[state] (p) {$p$};
    \node[state] (q) [right=5cm of p]{$q$};
    \path (p) edge node {$\varepsilon|\gamma|\push{1}\rew{a}\push{k+1}\pop{1}v$} (q);
\end{tikzpicture}
}

\begin{document}
\maketitle

\begin{abstract}
  The problem of \emph{regular separability} asks, given two languages
  $K$ and $L$, whether there exists a regular language $S$ with
  $K\subseteq S$ and $S\cap L=\emptyset$.  This problem has recently
  been studied for various classes of languages. All the results on
  regular separability obtained so far exhibited a noteworthy
  correspondence with the intersection emptiness problem: In each
  case, regular separability is decidable if and only if intersection
  emptiness is decidable. This raises the question whether under mild
  assumptions, regular separability can be reduced to intersection
  emptiness and vice-versa.
  
  We present counterexamples showing that none of the two problems can
  be reduced to the other. More specifically, we describe language
  classes $\calC_1$, $\calD_1$, $\calC_2$, $\calD_2$ such that
  (i)~intersection emptiness is decidable for $\calC_1$ and $\calD_1$,
  but regular separability is undecidable for $\calC_1$ and $\calD_1$
  and (ii)~regular separability is decidable for $\calC_2$ and
  $\calD_2$, but intersection emptiness is undecidable for $\calC_2$
  and $\calD_2$.
\end{abstract}

\section{Introduction} 
\label{sec:introduction}

The \emph{intersection emptiness problem} for language classes $\calC$
and $\calD$ asks for two given languages $K$ from $\calC$ and $L$ from
$\calD$, whether $K\cap L=\emptyset$. If $\calC$ and $\calD$ are
language classes associated to classes of infinite-state systems, ,then
intersection emptiness corresponds to verifying safety properties in
concurrent systems where one system of $\calC$ communicates with a
system of $\calD$ via messages or shared memory\cite{bouajjani2003generic}.  The
idea of separability
is to decide whether two given
languages are not
only disjoint, but whether there exists a finite, easily verifiable,
certificate for disjointness (and thus for safety).  Specifically, the
\emph{$\mathcal{S}$ separability problem} for a fixed class
$ \mathcal{S}$ of separators and language classes $\calC$ and $\calD$
asks, for given languages $K$ from $\calC$ and $L$ from $\calD$,
whether there exists a language $S \in \mathcal{S}$ with
$K\subseteq S$ and $S\cap L=\emptyset$.

There is extensive literature dealing with the separability problem,
with a range of different separators considered.  One line of work
concerns separability of regular languages by
separators from a variety of regular languages. Here, the
investigation began with a more general problem, \emph{computing
  pointlikes} (equivalently, the \emph{covering
  problem})~\cite{almeida1999,Gool2019a,DBLP:journals/lmcs/PlaceZ18},
but later also concentrated on separability
(e.g. \cite{DBLP:conf/lics/Place15,DBLP:conf/mfcs/PlaceRZ13,DBLP:conf/stacs/PlaceZ15,
  DBLP:journals/corr/PlaceZ14,DBLP:conf/csr/PlaceZ17,
  DBLP:conf/lics/PlaceZ17}). Moreover, separability has been studied
for regular tree languages, where separators are either piecewise
testable tree languages~\cite{DBLP:conf/icalp/Goubault-Larrecq16} or
languages of deterministic tree-walking
automata~\cite{DBLP:journals/fuin/Bojanczyk17}. For
non-regular input languages, separability has been investigated with
\emph{piecewise testable
  languages}~(PTL)~\cite{DBLP:journals/dmtcs/CzerwinskiMRZZ17} and
generalizations thereof~\cite{DBLP:conf/lics/Zetzsche18} as
separators.  Separability of subsets of trace
monoids~\cite{choffrut2006separability} and commutative
monoids~\cite{DBLP:conf/stacs/ClementeCLP17} by recognizable subsets
has been studied as well.

A natural choice for the separators is the class of \emph{regular
  languages}.  On the one hand, they have relatively high separation
power and on the other hand, it is usually verifiable whether a given
regular language is in fact a separator. For instance, they generalize
piecewise testable languages but are less powerful than context-free
languages (CFL).  Since the intersection problem for CFL is
undecidable, it is not easy to check if a given candidate CFL is a
separator.

This has motivated a recent research effort to understand for which
language classes $\calC,\calD$ regular separability is
decidable~\cite{lasota2019regular,DBLP:conf/stacs/ClementeCLP17,DBLP:conf/icalp/ClementeCLP17}. An early result was that regular
separability is undecidable for CFL (by this we
mean that both input languages are
context-free)~\cite{SzymanskiW-sicomp76,DBLP:journals/jacm/Hunt82a}. More
recently, it was shown that regular separability is undecidable
already for one-counter languages~\cite{DBLP:conf/lics/CzerwinskiL17},
but decidable for several subclasses of vector addition systems
(VASS): for one-dimensional VASS~\cite{lasota2019regular},
for commutative VASS languages~\cite{DBLP:conf/stacs/ClementeCLP17},
and for Parikh automata (equivalently,
$\Z$-VASS)~\cite{DBLP:conf/icalp/ClementeCLP17}. Moreover, it is
decidable for languages of well-structured transition
systems~\cite{DBLP:conf/concur/CzerwinskiLMMKS18}. Furthermore,
decidability still holds in many of these cases if one of the inputs
is a general VASS language~\cite{CzerwinskiZetzsche2019a}. However, if
both inputs are VASS languages, decidability of regular
separability remains a challenging open problem.

These results exhibit a striking correspondence between regular
separability and the intersection problem: In all the cases where
decidability of regular separability has been clarified, it is
decidable if and only if intersection is decidable.  In fact, in the
case of well-structured transition systems, it even turned out that
two languages are regularly separable if and only if they are
disjoint~\cite{DBLP:conf/concur/CzerwinskiLMMKS18}. Moreover, deciding
regular separability usually involves non-trivial refinements of the
methods for deciding intersection. Furthermore, so far the only method
to show undecidability of regular separability is to adapt
undecidability proofs for the intersection problem~\cite{DBLP:journals/jacm/Hunt82a,DBLP:conf/lics/CzerwinskiL17,DBLP:conf/lics/Zetzsche18}.

In light of these observations, there was a growing interest in
whether there is a deeper connection between regular separability and
intersection emptiness. In other words: \emph{Is regular separability
  just intersection emptiness in disguise?} It is conceivable that
under mild assumptions, regular separability and intersection
emptiness are mutually reducible.  An equivalence in this spirit
already exists for separability by PTL: If $\calC$ and $\calD$ are
closed under rational transductions, then separability by PTL for
$\calC$ and $\calD$ is decidable if and only if the simultaneous
unboundedness problem is decidable for $\calC$ and for
$\calD$~\cite{DBLP:journals/dmtcs/CzerwinskiMRZZ17}.

\subparagraph*{Contribution} We show that regular separability and
intersection emptiness are independent problems: Each problem can be
decidable while the other is undecidable. Specifically, we present
language classes $\calC_1$, $\calD_1$, $\calC_2$, $\calD_2$, so that
(i)~for $\calC_1$ and $\calD_1$, regular separability is undecidable,
but intersection emptiness is decidable and (ii)~for $\calC_2$ and
$\calD_2$, regular separability is decidable, but intersection
emptiness is undecidable. Some of these classes have been studied
before (such as the higher-order pushdown languages), but some have
not (to the best of our knowledge). However, they are all natural in
the sense that they are defined in terms of machine models, are closed
under rational transductions, and have decidable emptiness and
membership problems. We introduce two new classes defined by counter
systems that accept based on certain numerical predicates. These
predicates can be specified either using reset vector addition systems
or higher-order pushdown automata.



\section{Preliminaries} 
\label{sec:preliminaries}
We use $\Sigma$ (sometimes $\Gamma$) to denote a finite set of
alphabets and $\Sigma^*$ to denote the set of finite strings (aka
words) over the alphabet $\Sigma$.  To distinguish between expressions
over natural numbers and expressions involving words, we use
typewriter font to denote letters, e.g. $\ltr{a}$, $\ltr{0}$,
$\ltr{1}$, etc.  For example, $\ltr{0}^n$ is the word consisting of an
$n$-fold repetition of the letter $\ltr{0}$, whereas $0^n$ is the
number zero. The empty string is denoted $\varepsilon$. If $S\subseteq\N$
we write $\ltr{a}^S$ for the set $\{\ltr{a}^n \mid n\in S\}\subseteq \ltr{a}^*$
and $2^S$ for the set $\{2^n \mid n\in S\}\subseteq\N$.

We define the map $\nu \colon\{ \ltr{0},\ltr{1}\}^* \rightarrow \N$
which takes every word to the number which it represents in binary
representation: We define $\nu(\varepsilon)=0$ and
$\nu(w\ltr{1})=2\cdot\nu(w)+1$ and $\nu(w\ltr{0})=2\cdot \nu(w)$ for
$w\in\{\ltr{0},\ltr{1}\}^*$.  For example,
$\nu(\ltr{1}\ltr{1}\ltr{0})=6$.  Often we are only concerned with
words of the form $\{\ltr{0}\} \cup \ltr{1}\{\ltr{0},\ltr{1}\}^*$.

Languages are denoted by $L,L',K$ etc. and the language of a machine
$M$ is
denoted by $\langof{M}$. Classes of languages are denoted by
$\calC$, $\calD$, etc.

\begin{defn}
	\label{transducer}
	An \textit{asynchronous transducer} $\calT$ is a tuple $\calT=(
	Q,\Gamma,\Sigma,E,q_0,F)$
	with a set of finite states $Q$, finite output alphabet $\Gamma$,
	finite input
	alphabet $\Sigma$, a set of edges $E \subseteq Q \times
	\Gamma^*
	\times \Sigma^* \times Q$, initial state $q_0 \in Q$ and set of final
	states
	$F \subseteq Q$. We write $p \xrightarrow{v|u} q$ if $(p,u,v,q) \in
	E$ and the machine reads $v$ in state $p$, outputs $u$ and moves to
	state $q$. We also write $p \autstepps{w|w'} q$ if there are
	states $q_0,q_1 \cdots q_n$ and words $u_1,u_2, \cdots ,u_n, v_1,v_2,
	\cdots , v_n$ such that $p=q_0,q=q_n,w'=u_1u_2\cdots u_n,
	w=v_1v_2\cdots v_n$ and $q_i \xrightarrow{v_i|u_i}
	q_{i+1}$ for all $0 \leq i \leq n$.

	The \textit{transduction} $T \subseteq \Gamma^* \times
	\Sigma^*$ generated by the transducer $\calT$ is the set of tuples
	$(u,v) \in \Gamma^* \times \Sigma^*$ such that $q_0 \autstepps{v|u}
	q_f$ for some $q_f \in F$. Given a language $L \subseteq \Sigma^*$,
	we define
	$TL:=\{ u \in \Gamma^* \; | \; \exists v \in L \; (u,v) \in T\}$.
        A transduction $T\subseteq\Gamma^*\times\Sigma^*$ is \emph{rational} if
        it is generated by some asynchronous transducer.
\end{defn}

A \emph{language} is a subset of $\Sigma^*$ for some alphabet
$\Sigma$.  A \emph{language class} is a collection of languages,
together with some way to finitely represent these languages, for
example using machine models or grammars. We call a language class a
\emph{full trio} if it is effectively closed under rational
transductions.  This means, given a representation of $L$ in $\calC$
and an asynchronous transducer for
$T\subseteq \Sigma^*\times\Gamma^*$, the language $TL$ belongs to
$\calC$ and one can compute a representation of $TL$ in $\calC$.

The following equivalent definition of full trios is well known 
(see Berstel \cite{berstel2013transductions}):
\begin{lem}
	\label{alt_fulltrio}
	A language class is closed under rational transductions if and
        only if it is closed under (i)~homomorphic image, (ii)~inverse
        homomorphic image, and (iii)~intersection with regular
        languages.
\end{lem}

We are interested in decision problems where the representation of a
language $L$ (or possibly multiple languages) is the input. In
particular, we study the following problems.
\begin{prob}[Intersection Emptiness]
	\label{prob:inter}
	Given two languages classes $\calC_1$ and $\calC_2$, the intersection
	emptiness
	problem $\Inter{\calC_1}{\calC_2}$ is defined as follows:
        \begin{description}
	\item[Input:] Languages $L_1 \in \calC_1$ and
	$L_2 \in \calC_2$.
      \item[Question:] Is $L_1 \cap L_2$ empty?
        \end{description}
\end{prob}

\begin{prob}[Regular Separability]
  \label{prob:reg_sep}
  Given two languages classes $\calC_1$ and $\calC_2$, the regular
  separability
  problem $\RegSep{\calC_1}{\calC_2}$ is defined as follows:
  \begin{description}
  \item[Input:] Languages $L_1 \in \calC_1$ and
    $L_2 \in \calC_2$.
  \item[Question:] Is there a regular language $R$ such that $L_1
    \subseteq R$ and $L_2 \cap R=\emptyset$?
  \end{description}
\end{prob}
We will write $\regsep{L}{K}$ to denote that $L$ and $K$ are regularly
separable. 

\begin{prob}[Emptiness]
	\label{prob:emptiness}
	The emptiness problem for a language class $\calC$, denoted $\Empty
	{\calC}$ is defined as:
        \begin{description}
	\item[Input:] A language $L \in \calC$.
        \item[Question:] Is $L=\emptyset$ i.e. is $L$ empty?
        \end{description}
\end{prob}

\begin{prob}[Infinity]
	\label{prob:infinity}
	The infinity problem for a language class $\calC$, denoted $\Inf
	{\calC}$ is defined as:
        \begin{description}
	\item[Input:] A language $L \in \calC$.
        \item[Question:] Does $L$ contain infinitely many elements?
        \end{description}
\end{prob}

\section{Incrementing automata}
The counterexamples we construct are defined using special kinds of
automata which can only increment a counter, which we will define
formally below. The acceptance condition requires that the counter
value satisfy a specific \textit{numerical predicate}, in addition to
reaching a final state. By a \emph{predicate class}, we mean a class
$\calP$ of predicates over natural numbers (i.e. subsets
$P\subseteq\N$) such that there is a way to finitely describe the
members of $\calP$.  As an example, if $\calC$ is a language class,
then a subset $S \subseteq \N$ is a \textit{pseudo-$\calC$ predicate}
if $S=\nu(L)$ for some $L \in \calC$. Now the class of all
pseudo-$\calC$ predicates constitutes a predicate class, because a
pseudo-$\calC$ predicate can be described using the finite description
of a language in $\calC$.  The class of all pseudo-$\calC$ predicates
is denoted $\pseudo{\calC}$.

\begin{defn}
  \label{defn:special_counter_machines} 
  Let $\calP$ be a predicate class.  An \emph{incrementing automata
    over $\calP$} is a four-tuple $\calM=(Q,\Sigma,E,q_0,F)$ where $Q$
  is a finite set of \emph{states}, $\Sigma$ is its \emph{input
    alphabet}, $E \subseteq Q \times \Sigma^* \times \{0,1\} \times Q$
  a finite set of \emph{edges}, $q_0\in Q$ an initial state and $F$ is
  a finite set of \emph{acceptance pairs} $(q,P)$ where $q \in Q$ is a
  state and $P$ belongs to $\calP$ .

  A \emph{configuration} of $\calM$ is a pair $(q,n) \in Q \times \N$.
  For two configurations $(q,n)$, $(q',n')$, we write
  $(q,n)\autstepp{w}(q',n')$ if there are configurations
  $(q_1,n_1),\ldots,(q_\ell,n_\ell)$ with $q_1=q$ and $q_\ell=q'$ and
  edges $(q_{i},w_i,m_i,q_{i+1})$ with $n_{i+1}=n_i+m_{i+1}$ for
  $1\le i<\ell$ and $w=w_1\cdots w_\ell$.  The \emph{language accepted
    by $\calM$} is
  \[ \langof{\calM}=\{w\in\Sigma^* \mid \text{$(q_0,0)\autstepp{w} (q,m)$ for some $(q,P)$ in $F$ with $m\in P$}\}. \]
  The collection of all languages accepted by incrementing automata over $\calP$
  is denoted $\cntlang{\calP}$. 
  \end{defn}

  It turns out that even with no further assumptions on the predicate
  class $\calP$, the language class $\cntlang{\calP}$ has some nice
  closure properties.
	\begin{lem}
          \label{lem:pseudoCtMc_rat_trans}
          Let $\calP$ be a predicate class. The languages of
          incrementing automata over $\calP$ are precisely the finite unions of
          languages of the form $T\ltr{a}^P$ where $P \in \calP$ and
          $T\subseteq\Sigma^*\times\{\ltr{a}\}^*$ is a rational
          transduction. In particular, the class of languages accepted
          by incrementing automata over $\calP$ is a full trio.
	\end{lem}
	\begin{proof}
          For every accepting pair $(q,P)$ of $\cal M$, we construct a
          transducer $T_{q,P}$, which has the same set of states as
          $\cal M$, accepting state set $\{q\}$ and for each edge
          $ (q',w,m,q'')$ of $\calM$ the transducer reads $\ltr{a}$ if
          $m=1$ or $\varepsilon$ if $m=0$ and outputs $w$. Then
          $\langof{\cal M}$ is the finite union of all
          $T_{q,P}(\ltr{a}^P)$.

          Conversely, since the languages accepted by incrementing automata over $\calP$ are clearly closed under union, it suffices to show
          that $T\ltr{a}^P$ is accepted by an incrementing
          automaton over $\calP$. We may assume that $T$ is given by a
          transducer in
          which every edge is of the form $(q,w,\ltr{a}^m,q')$ with
          $m\in\{0,1\}$. Let $\calM$ have the same state set as $T$
          and turn every edge $(q,w,\ltr{a}^m,q')$ into an edge
          $(q,w,m,q')$ for $\calM$. Finally, for every final state
          $q$ of $T$, we give $\calM$ an accepting pair $(q,P)$. Then
          clearly $\langof{\calM}=T\ltr{a}^P$.

          This implies that the class of incrementing automata over $\calP$ is a
          full trio: If $L\subseteq\Sigma^*$ is accepted by a
          incrementing automata over $\calP$, then we can write
          $L=T_1\ltr{a}^{P_1}\cup\cdots\cup T_\ell\ltr{a}^{P_\ell}$
          with $T_1,\ldots,T_\ell\subseteq\Sigma^*\times\ltr{a}^*$. If
          $T\subseteq\Gamma^*\times\Sigma^*$ is a rational
          transduction, then
          $TL=(TT_1)\ltr{a}^{P_1}\cup\cdots\cup
          (TT_\ell)\ltr{a}^{P_\ell}$ and since $TT_i$ is again a
          rational transduction for $1\le i\le \ell$, the lanuage $TL$
          is accepted by some incrementing automata over $\calP$.
	\end{proof}

        It is obvious that the class $\cntlang{\calP}$ does not always
        have a decidable emptiness problem: Emptiness is decidable for
        $\cntlang{\calP}$ if and only if it is decidable
        whether a given predicate from $\calP$ intersects a given
        arithmetic progression, i.e. given $P$ and $m,n\in\N$, whether
        $(m+n\N)\cap P\ne\emptyset$. For all the predicate classes
        $\calP$ we consider, emptiness for $\cntlang{\calP}$ will
        always be decidable.

\section{Decidable Intersection and Undecidable Regular Separability}
\label{sec:decidable_intersection_and_undecidable_regular_separability}

In this section, we present a language class $\calC$ so that
intersection emptiness problem $\Inter{\calC}{\calC}$ is decidable for
$\calC$, but the regular
separability problem $\RegSep{\calC}{\calC}$ is undecidable for
$\calC$.
The definition of $\calC$
is based on reset vector addition systems.


\subparagraph*{Reset Vector Addition Systems} 
\label{reset-vass}
%
%

%
  A \emph{reset vector addition system} (reset VASS) is a tuple
  $\calV =(Q,\Sigma,n,E,q_0,F)$, where $Q$ is a finite set of \emph{states},
  $\Sigma$ is its \emph{input alphabet}, $n\in\N$ is its number of
  counters,
  $E\subseteq Q\times\Sigma^*\times \{1,\ldots,n\}\times\{0,1,-1,\reset\}\times Q$ is a
  finite set of \emph{edges}, $q_0\in Q$ is its \emph{initial state},
  and $F\subseteq Q$ is its set of \emph{final states}.  A
  \emph{configuration} if $\calV$ is a tuple $(q,m_1,\ldots,m_n)$
  where $q\in Q$ and $m_1,\ldots,m_n\in\N$. We write $(q,m_1,\ldots,m_n)\autstepp{w}(q',m'_1,\ldots,m'_n)$ if there is an edge $(q,w,k,x,q')$ such that for every $j\ne k$, we have $m'_{j}=m_{j}$ and
    \begin{itemize}
    \item if $x\in\{-1,0,1\}$, then $m'_k=m_k+x$,
    \item if $x=\reset$, then $m'_k=0$.
    \end{itemize}
  If there are configurations $c_1,\ldots,c_\ell$ and words
  $w_1,\ldots,w_{\ell-1}$ with $c_i\autstepp{w_i}c_{i+1}$ for
  $1\le i<\ell$, and $w=w_1\cdots w_\ell$, then we also
  write $c_1\autstepp{w}c_\ell$. The 
  \emph{language accepted by $\calV$} is defined as
  \[ \langof{\calV} = \{w\in\Sigma^* \mid \text{$(q_0,0,\ldots,0)\autstepp{w}(q,m_1,\ldots,m_n)$ for some $q\in F$ and $m_1,\ldots,m_n\in\N$}\}. \]
  The class of languages accepted by reset VASS is denoted $\calR$.

 Our language class will be $\cntlang{\pseudo{\calR}}$,
 i.e. incrementing automata with access to predicates of the form
 $\nu(L)$ where $L$ is the language of a reset VASS.
 \begin{thm}\label{undecidable-sep}
   $\RegSep{\cntlang{\pseudo{\calR}}}{\cntlang{\pseudo{\calR}}}$ is
     undecidable and
     $\Inter{\cntlang{\pseudo{\calR}}}{\cntlang{\pseudo{\calR}}}$ is
     decidable.
 \end{thm}
 Note that $\cntlang{\pseudo{\calR}}$ is a full trio
 (\cref{lem:pseudoCtMc_rat_trans}) and since intersection is
 decidable, in particular its emptiness problem is decidable: One has
 $L\cap\Sigma^*=\emptyset$ if and only if $L=\emptyset$. Moreover,
 note that we could not have chosen $\calR$ as our example class:
 Since reset VASS are well-structured transition systems, regular
 separability is decidable for
 them~\cite{DBLP:conf/concur/CzerwinskiLMMKS18}.

 Before we begin with the proof of \cref{undecidable-sep}, let us
 mention that instead of $\calR$, we could have chosen any language
 class $\calD$, for which
 (i)~$\calD$ is closed under rational transductions,
 (ii)~$\calD$ is closed under intersection,
 (iii)~$\Empty{\calD}$ is decidable and
 (iv)~$\Inf{\calD}$ is undecidable.
 For example, we could have also used lossy channel systems instead of
 reset VASS. 
 
We now recall some results regarding $\calR$ from literature. 
\begin{lem}
	\label{lem:emptyLCMdec}
	Emptiness is decidable for $\calR$.
\end{lem}
The \lcnamecref{lem:emptyLCMdec} follows from the fact that reset VASS
are well-structured transition
systems~\cite{DBLP:conf/icalp/DufourdFS98}, for which the coverability
problem is
decidable~\cite{abdulla1996general,DBLP:conf/latin/FinkelS98} and the
fact that a reset VASS has a non-empty language if and only if a particular
configuration is coverable.

The following can be shown using standard product constructions,
please see \cref{proof:reset_clo}.
\begin{lem}
	\label{thm:LCM_closure_rat_trans}
	$\calR$ is closed under rational transductions, union, and intersection.
\end{lem}

We now show that regular separability is undecidable for
$\cntlang{\pseudo{\calR}}$. We do this using a reduction from the
infinity problem for $\calR$, whose undecidability is an easy
consequence from the undecidability of boundedness of reset VASS.

The boundedness problem for reset VASS is defined below and was shown
to be undecidable by Dufourd, Finkel, and
Schnoebelen~\cite{DBLP:conf/icalp/DufourdFS98} (and a simple and more
general proof was given by Mayr~\cite{mayr2003undecidable}).  A
configuration $(q,x_1,\ldots,x_n)$ is \emph{reachable} if there is a
$w\in\Sigma^*$ with
$(q_0,0,\ldots,0)\autstepp{w}(q,x_1,\ldots,w_n)$. A reset VASS $\calV$
is called \emph{bounded} if there is a $B \in \N$ such that for every
reachable $(q,x_1,\ldots,x_n)$, we have $x_1+\cdots+x_n\le B$. Hence,
the \emph{boundedness problem} is the following.
        \begin{description}
	\item[Input:] A reset VASS $\calV$.
	\item[Question:] Is $\calV$ bounded?
          \end{description}

\begin{lem}
\label{lem:infLCMundec}
	The infinity problem for $\calR$ is undecidable.
\end{lem}
\begin{proof}
  From an input reset VASS $\calV=(Q,\Sigma,n,E, q_0, F)$ ,
  we construct a reset VASS $\calV'$ over the alphabet
  $\Sigma'=\{ \ltr{a} \}$ as follows.  In every edge of $\calV$, we
  replace the input word by the empty word $\varepsilon$.  Moreover,
  we add a fresh state $s$, which is the only final state of $\calV'$.
  Then, we add an edge $(q,\varepsilon,1,0,s)$ for every state $q$ of
  $\calV$.  Finally, we add a loop $(s,\ltr{a},i,-1,s)$ for every
  $i\in\{1,\ldots,n\}$.
  This means $\calV'$ simulates a computation of $\calV$ (but
  disregarding the input) and can spontaneously jump into the state
  $s$, from where it can decrement counters.  Each time it decrements
  a counter in $s$, it reads an $\ltr{a}$ from the input.  Thus,
  clearly, $\langof{\calV'}\subseteq\ltr{a}^*$. Moreover, we have
  $\ltr{a}^m\in\langof{\calV'}$ if and only if there is a reachable
  configuration $(q,x_1,\ldots,x_n)$ of $\calV$ with
  $x_1+\cdots+x_n\ge m$. Thus, $\langof{\calV'}$ is finite if and only
  if $\calV$ is bounded.
\end{proof}
Note that infinity is already undecidable for languages that are
subsets of $\ltr{1}\ltr{0}^*$. This is because given $L$ from
$\calR$, a rational transduction yields
$L'=\{\ltr{1}\ltr{0}^{|w|} \mid w\in L\}$ and $L'$ is infinite
if and only if $L$ is.

Our reduction from the infinity problem works because the input
languages have a particular shape, for which regular separability has
a simple characterization.
\begin{lem}
\label{lem:unary_reg_sep_pow2}
	Let $S_0,S_1 \subseteq \N$ and $\N \setminus 2^{\N} \subseteq
	S_1$. Then $\ltr{a}^{S_0}$ and $\ltr{a}^{S_1}$ are regularly separable if and
	only if $S_0$ is finite and disjoint from $S_1$.
\end{lem}
\begin{proof}
  If $S_0$ is finite and disjoint from $S_1$, then clearly $\ltr{a}^{S_0}$ is a
  regular separator. For the only if direction, consider any infinite
  regular language $R \subseteq \ltr{a}^*$. It has to include an
  arithmetic progession, meaning that there exist $m,n\in\N$ with
  $\ltr{a}^{m + n\N} \subseteq R$. Hence, for sufficiently large $\ell$, the language
  $ \{ \ltr{a}^x\mid 2^\ell < x < 2^{\ell+1} \} \subseteq S_1$ must
  intersect with $R$. In other words, no infinite $R$ can be a regular
  separator of $\ltr{a}^{S_0}$ and $\ltr{a}^{S_1}$ i.e.  $S_0$ must be
  finite (and disjoint from $S_1$).
\end{proof}

\begin{lem}
	\label{thm:undec_reg_sep_pseudoLCM}
	Regular separability is undecidable for $\cntlang{\pseudo{\calR}}$.
\end{lem}
\begin{proof}
  We reduce the infinity problem for $\calR$ (which is undecidable by
  \cref{lem:infLCMundec}) to regular separability in
  $\cntlang{\pseudo{\calR}}$. Suppose we are given $L$ from $\calR$.
  Since $\calR$ is effectively closed under rational transductions, we
  also have $K=\{\ltr{1}\ltr{0}^{|w|} \mid w\in L\}$ in $\calR$.  Note
  that $K$ is infinite if and only if $L$ is infinite.  Then
  $\nu(K)\subseteq 2^\N$ and $K_1:=\ltr{a}^{\nu(K)}$ belongs to
  $\cntlang{\pseudo{\calR}}$.  Let
  $K_2=\ltr{a}^{\N\setminus 2^{\N}}=\ltr{a}^{\nu(\ltr{1}
    \{\ltr{0},\ltr{1}\}^*\ltr{1}\{\ltr{0},\ltr{1}\}^*)}$, which also
  belongs to $\cntlang{\pseudo{\calR}}$, because
  $\ltr{1}\{\ltr{0},\ltr{1}\}^*\ltr{1}\{\ltr{0},\ltr{1}\}^*$ is
  regular and thus a member of $\calR$.
  
  By Lemma \ref{lem:infLCMundec}, $K_1$ and $K_2$ are regularly
  separable if and only if $K_1$ is finite and disjoint from $K_2$.
  Since $K_1 \cap K_2 = \emptyset$ by construction, we have regular
  separability if and only if $K_1$ is finite, which happens if and
  only if $K$ is finite.
\end{proof}

For \cref{undecidable-sep}, it remains to show that intersection is
decidable for
$\cntlang{\pseudo{\calR}}$.  We do this by expressing intersection
non-emptiness in the logic $\Sigma_1^+(\N,+,\le,1,\pseudo{\calR})$, which is
the positive $\Sigma_1$ fragment of Presburger arithmetic extended
with pseudo-$\calR$ predicates. Moreover, we show that this logic has
a decidable truth problem.

We begin with some notions from first-order logic (please
see~\cite{enderton2001mathematical} for syntax and semantics of
first-order logic). First-order formulae will be denoted by
$\phi(\bar{x}),\psi(y)$ etc. where $\bar{x}$ is a tuple of (possibly
superset of the) free variables and $y$ is a single free variable.
For a formula $\phi(\bar{x})$, we denote by
$\llbracket \phi(\bar{x}) \rrbracket$ the set of its solutions (in
our case, the domain is $\N$).

\newcommand{\structure}{\ensuremath{(\N,+,\le,1,\pseudo{\calR})}\xspace}
\newcommand{\logic}{\ensuremath{\Sigma_1^+\structure}\xspace}
\newcommand{\pifragment}{\ensuremath{\Pi_1^+\structure}\xspace}

Our decision procedure for \logic is
essentially the same as the procedure to decide the first-order theory of
automatic structures~\cite{blumensath2000automatic}, except that
instead of regular languages, we use $\calR$.  For
$\bar{w}=(w_1,w_2,\ldots,w_k) \in (\Sigma^*)^k$, the
\textit{convolution} $w_1 \otimes w_2 \otimes \ldots \otimes w_k$ is a
word over the alphabet $(\Sigma \cup \{\Box\})^k$ where $\Box$ is a
padding symbol not present in $\Sigma$. If
$w_i=w_{i1} w_{i2} \ldots w_{im_i}$ and
$m=max \{ m_1,m_2,\ldots,m_k\}$ then
\[ w_1 \otimes w_2 \otimes \ldots \otimes w_k :=\begin{bmatrix}
w'_{11}\\
w'_{21} \\
\vdots \\
w'_{k1}
\end{bmatrix}
\hdots \begin{bmatrix}
w'_{1m}\\
w'_{2m} \\
\vdots \\
w'_{km}
\end{bmatrix} \in ((\Sigma \cup \{\Box\})^k)^*\] where
$w'_{ij}=w_{ij}$ if $j \leq m_i$ else $w'_{ij}=\Box$.  The
\textit{reversal} $\rev{w}$ of a word $w=w_1w_2w_3\cdots w_n$ (where
each $w_i$ is a letter) is $\rev{w}=w_nw_{n-1}\cdots w_1$.
We say that a $k$-ary (arithmetic) relation $R \subseteq \N^k$ is a
\emph{pseudo-$\calR$ relation} if the set of words $L_R=\{ \rev{w_1}
\otimes \rev{w_2} \otimes
\cdots
\otimes
\rev{w_k} \; | \; (\nu(w_1),\ldots,\nu(w_k)) \in R \}$ \footnote{This
definition
is stated using the reversal purely for technical purposes: the
addition relation is automatic on the reverse convolution.} belongs to
$\calR$. In our decision procedure for \logic, we will show inductively that
every formula defines a pseudo-$\calR$ relation.

Formally, we consider the theory $\logic$ where \structure is the
structure with domain $\N$ of natural numbers, the constant symbol $1$
and the binary symbols $+$ and $\le$ taking their canonical
interpretations and $\pseudo{\calR}$ is a set of predicate symbols,
one for each $\pseudo{\calR}$ predicate. By $\Sigma^+_1$ we mean the
fragment of first order formulae obtained by using only the boolean
operations $\wedge,\vee$ and existential quantification.
\begin{defn}
\label{defn:posExist_presArith_LCM}
		Let $\logic$ be the set of
	first order
	logic
	formulae given
	by the following grammar:
	\begin{flalign*}
 \phi(\bar{x},\bar{y},\bar{z}) := & S(x) \; | \; t_1\le t_2 \; | \;
\phi_1(\bar{x},\bar{y}) \wedge \phi_2(\bar{x},\bar{z}) \; | \; \phi_1(
\bar{x},\bar{y}) \vee \phi_2(\bar{x},\bar{z}) \; | \; \exists y \;
\phi'(y,\bar{x})
\end{flalign*}
where $S \in \pseudo{\calR}$ and $t_1,t_2$ are terms obtained from using
variables,
$1$ and $+$.
\end{defn}

\begin{lem}
	\label{lem:posExist_presArithLCM_truth_dec}
	The truth problem for $\logic$ is decidable.
\end{lem}
\begin{proof}
It is clear that by introducing new existentially quantified
variables, one can transform each formula from $\logic$ into an
equivalent formula that is generated by the simpler grammar
\begin{flalign*}
 \phi(\bar{x},\bar{y},\bar{z}) := & S(x) \; | \;
x+y=z \; | \;
x=1\; | \\
&\phi_1(\bar{x},\bar{y}) \wedge \phi_2(\bar{x},\bar{z}) \; | \; \phi_1(
\bar{x},\bar{y}) \vee \phi_2(\bar{x},\bar{z}) \; | \; \exists y \;
\phi'(y,\bar{x})
\end{flalign*}
We want to show that given any input sentence $\psi$ from $\logic$, we
can decide if it is true or not. If the sentence has no variables,
then it is trivial to decide. Otherwise,
$\psi=\exists \bar{x} \; \phi(\bar{x})$ for some formula
$\phi( \bar{x})$. We claim that the solution set
$R=\llbracket \phi( \bar{x}) \rrbracket$ belongs to $\pseudo{\calR}$
and a reset VASS for $L_R$ can be effectively computed. Assuming the
claim, the truth of $\psi$ reduces to the emptiness of
$\llbracket \phi( \bar{x}) \rrbracket$ or equivalently the emptiness
of $L_R$, which is decidable by Lemma \ref{lem:emptyLCMdec}.

	We prove the claim by structural induction on the defining formula
	$\phi(
	\bar{x})$, please see \cref{proof:truth_problem} for details.
\end{proof}

\begin{rem}
	\label{undec-pi_1}
	The truth problem for $\pifragment$ is undecidable by
        reduction from the infinity problem for
        $\calR$. Given $L \subseteq \ltr{1}\ltr{0}^*$, let
        $R_L=\nu(L) \subseteq \N$ be the predicate corresponding to
        $L$. Now the downward closure
        $D:= \{ x \in \N \mid \exists y\colon x \leq y \wedge R_L(y)\}$
        is definable in $\logic$ and therefore we have
        $K \in \calR$ with $K=\nu(D)$ by the proof of
        \cref{lem:posExist_presArithLCM_truth_dec}. Then the $\Pi_1^+$-sentence
        $\forall x\colon R_K(x)$ is true if and only if $L$ is
        infinite.
\end{rem}

Having established that $\logic$ is decidable, we are ready to show
that intersection emptiness is decidable for
$\cntlang{\pseudo{\calR}}$.
\begin{lem}
	The intersection problem is decidable for $\cntlang{\pseudo{\calR}}$. 
\end{lem}
\begin{proof}
  Given $L_1,L_2 \in \cntlang{\pseudo{\calR}}$, by Lemma
  \ref{lem:pseudoCtMc_rat_trans}, we know that both $L_1$ and $L_2$
  are finite unions of languages of the form $T\ltr{a}^S$, where $S$ is a pseudo-$\calR$ predicate. Therefore, it suffices to decide the emptiness of intersections of the form
  $T_1\ltr{a}^{S_1} \cap T_2\ltr{a}^{S_2}$ where $S_1$ and $S_2$ are
  pseudo-$\calR$ predicates. Note that
  $T_1\ltr{a}^{S_1} \cap T_2\ltr{a}^{S_2} = \emptyset$ iff
  $T_2^{-1}T_1\ltr{a}^{S_1} \cap \ltr{a}^{S_2} = \emptyset$. Since $T_2^{-1}T_1$ is again a rational transduction, it
  suffices to check emptiness of languages of the form
  $T\ltr{a}^{S_1} \cap \ltr{a}^{S_2}$ where
  $T \subseteq \ltr{a}^* \times \ltr{a}^*$ is a rational transduction.
  Notice that we can construct an automaton $\mathcal{A}$ over the
  alphabet $\Sigma'=\{b,c\}$ with the same states as the transducer
  $\mathcal{M}_T$
  for $T$ and where for any transition $p \xrightarrow{\ltr{a}^m|\ltr{a}^n} q$ of
  $\mathcal{M}_T$ we have a transition $p \xrightarrow{\ltr{b}^m\ltr{c}^n} q$ in
  $\mathcal{A}$. It is clear that $(\ltr{a}^x,\ltr{a}^y) \in T$ iff there exists a
  word $w \in \langof{\mathcal{A}}$ such that $w$ contains exactly $x$ occurrences of $\ltr{b}$
  and $y$ occurrences of $\ltr{c}$.  
  Now it follows from Parikh's theorem \cite{parikh1966context} that the set $\{
  (x,y)\in\N\times\N \mid (\ltr{a}^x,\ltr{a}^y)\in T\}$ is semilinear, meaning that there are numbers $n_0,\ldots,n_k$ and $m_0,\ldots,m_k$ such that
  $(\ltr{a}^x,\ltr{a}^y)\in T$ if and only if
  \[ \exists z_1 \exists z_2 \ldots \exists z_k \; (x= n_0 + \sum_
    {i=i}^k z_in_i) \; \wedge \; (y= m_0 + \sum_
    {i=i}^k z_im_i).\]
  In particular, there is a formula $\phi_T(x,y)$ in  $\logic$
  such that $(\ltr{a}^x,\ltr{a}^y)\in T$ if and only if $\phi_T(x,y)$ is satisfied.
  We can
	now write a formula $\phi_2(y)$ in
	$\logic$ such that
        $\phi_2(y)$ is satisfied if and only if $\ltr{a}^y\in T\ltr{a}^{S_2}$:
	\[ \phi_2(y) := \exists x \; \phi_T(x,y) \; \wedge \; S_2(x)\]
        In the same way, the formula $\phi_1(x):=S_1(x)$ defines $\ltr{a}^{S_1}$.        
        Now define $\phi=\exists x \; \phi_1(x) \; \wedge \; \phi_2(x)$. Then clearly
$\phi \text{ is true }$ if and only if $T\ltr{a}^{S_2} \cap \ltr{a}^{S_1} \neq \emptyset$.
Decidability of $\Inter{\cntlang{\pseudo{\calR}}}{\cntlang{\pseudo{\calR}}}$ follows from 
\cref{lem:posExist_presArithLCM_truth_dec}.
\end{proof}





\section{Decidable Regular Separability and Undecidable Intersection}
\label{sec:decidable_regular_separability_and_undecidable_intersection}

In this section, we present language classes $\calC$ and $\calD$ so that
$\Inter{\calC}{\calD}$ is undecidable, but $\RegSep{\calC}{\calD}$ is decidable.
These classes are constructed using higher-order pushdown automata, which we define first.

We follow the definition of~\cite{HagueKochemsOng2016}. Higher-order pushdown automata are a generalization of pushdown automata where
instead of manipulating a stack, one can manipulate a stack of stacks
(order-$2$), a stack of stacks of stacks (order-$3$), etc.
Therefore, we begin by defining these higher-order stacks. While for ordinary (i.e. order-$1$) pushdown
automata, stacks are words over the stack alphabet $\Gamma$, order-$(k+1)$
stacks are sequences of order-$k$ stacks.
Let $\Gamma$ be an alphabet and $k\in\N$. The set of \emph{order-$k$
  stacks} $\stacks{\Gamma}{k}$ is inductively defined as follows:
\[ \stacks{\Gamma}{0}=\Gamma, ~~~~\stacks{\Gamma}{k+1}=\{[s_1\cdots s_m]_{k+1} \mid m\ge 1,~s_1,\ldots,s_m\in\stacks{\Gamma}{k} \}.\]
For a word $v\in\Gamma^+$, the stack $[\cdots[[v]_1]_2\cdots]_k$ is
also denoted $\llbracket v\rrbracket_k$. The function $\tops$ yields the topmost
symbol from $\Gamma$. This means, we have
$\tops([s_1\cdots s_m]_1)=s_m$ and $\tops([s_1\cdots s_m]_k)=\tops(s_m)$ for $k>1$.

Higher-order pushdown automata operate on higher-order stacks by way of
instructions. For the stack alphabet $\Gamma$ and for order-$k$ stacks, we have
the instruction set
$\instr{\Gamma}{k}=\{\push{i}, \pop{i}\mid 1\le i\le k\}\cup \{\rew{\gamma}\mid
\gamma\in\Gamma\}$. These instructions act on $\stacks{\Gamma}{k}$ as follows:
\begin{alignat*}{3}
  &[s_1\cdots s_m]_1 \cdot \rew{\gamma}&&=[s_1 \cdots s_{m-1}\gamma]_1 &\quad& \\
  &[s_1\cdots s_m]_k\cdot \rew{\gamma}&&=[s_1 \cdots s_{m-1}(s_m\cdot\rew{\gamma})]_k && \text{if $k>1$} \\
  &[s_1\cdots s_m]_i\cdot \push{i} &&= [s_1\cdots s_m s_m]_i & &  \\
  &[s_1\cdots s_m]_k\cdot \push{i} &&= [s_1\cdots s_m~(s_m\cdot\push{i})]_k & & \text{if $k>i$} \\
  &[s_1 \cdots s_m]_i\cdot \pop{i}&&= [s_1\cdots s_{m-1}]_i & & \text{if $m\ge 2$} \\
  &[s_1\cdots s_m]_k \cdot \pop{i}&&= [s_1\cdots s_{m-1}~(s_m\cdot\pop{i})]_k & & \text{if $k>i$}
\end{alignat*}
and in all other cases, the result is undefined.
For a word
$w\in (\instr{\Gamma}{k})^*$ and a stack $s\in\stacks{\Gamma}{k}$, the
stack $s\cdot w$ is defined inductively by $s\cdot\varepsilon=s$ and
$s\cdot (wx)=(s\cdot w)\cdot x$ for $x\in\instr{\Gamma}{k}$.

An \emph{(order-$k$) higher-order pushdown automaton} (short HOPA) is a tuple
$\calA=(Q,\Sigma,\Gamma,\bot,E,q_0,F)$, where $Q$ is a finite set of
\emph{states}, $\Sigma$ is its \emph{input alphabet}, $\Gamma$ is its
\emph{stack alphabet}, $\bot\in\Gamma$ is its \emph{stack bottom symbol}, $E$ is
a finite subset of
$Q\times\Sigma^*\times\Gamma\times(\instr{\Gamma}{k})^*\times Q$ whose elements
are called \emph{edges}, $q_0\in Q$ is its \emph{initial state}, and
$F\subseteq Q$ is its set of \emph{final states}. A \emph{configuration} is a
pair $(q,s)\in Q\times\stacks{\Gamma}{k}$. When drawing a higher-order pushdown automaton,
an edge $(q,u,\gamma,v,q')$ is represented by an arc $q\xrightarrow{u|\gamma|v}q'$.
An arc $q\xrightarrow{u|v}q'$ means that for each $\gamma\in\Gamma$, there is an edge
$(q,u,\gamma,v,q')$.

For configurations $(q,s),(q',s')$ and a word $u\in\Sigma^*$, we write
$(q,s)\autstepp[\calA]{u}(q',s')$ if there are edges
$(q_1,u_1,\gamma_1,v_1,q_2),(q_2,u_2,\gamma_2,v_2,q_3),\ldots,(q_{n-1},u_{n-1},\gamma_{n-1},v_{n-1},q_n)$
in $E$ and stacks $s_1,\ldots,s_n\in\stacks{\Gamma}{k}$ with
$\tops(s_i)=\gamma_i$ and $s_i\cdot v_{i}=s_{i+1}$ for $1\le i\le n-1$
such
that $(q,s)=(q_1,s_1)$ and $(q',s')=(q_n,s_n)$ and $u=u_1\cdots u_n$.  The
\emph{language accepted by $\calA$} is defined as
\[ \langof{\calA}=\{w\in\Sigma^* \mid (q_0,\llbracket\bot\rrbracket_k)\autstepp[\calA]{w}(q,s)~\text{for some $q\in F$ and $s\in\stacks{\Gamma}{k}$}\}.\]
The languages accepted by order-$k$ pushdown automata are called \emph{order-$k$ pushdown languages}. By $\calH$, we denote the class of languages accepted by an order-$k$ pushdown automaton
for some $k\in\N$.
In our example of classes with decidable regular separability and undecidable
intersection, one of the two classes is $\calH$. The other class will again be
defined using incrementing automata.

\begin{defn}
  \label{defn:paraPDlang} Let $\calC$ be a language class. A predicate
  $P\subseteq\N$ is a \emph{power-$\calC$ predicate} if
  $P=\N\setminus 2^{\N}~\cup~\{2^{\nu(w)}\mid w\in L\}$ for some language $L$
  from $\calC$.  The class of power-$\calC$ predicates is denoted $\power{\calC}$.
\end{defn}

Our example of classes with decidable regular separability but
undecidable intersection is $\calH$ on the one hand and
$\cntlang{\power{\calH}}$ on the other hand.
\begin{thm}\label{undecidable-inter}
  $\RegSep{\calH}{\cntlang{\power{\calH}}}$ is decidable, whereas
$\Inter{\calH}{\cntlang{\power{\calH}}}$ is undecidable. 
\end{thm}
Note that decidable regular separability implies that
$\cntlang{\power{\calH}}$ has a decidable emptiness problem: For
$L\subseteq\Sigma^*$, one has $\regsep{\Sigma^*}{L}$ if and only if
$L=\emptyset$. Moreover, note that we could not have chosen $\calH$ as
our counterexample, because regular separability is undecidable for
$\calH$ (already for context-free
languages)~\cite{SzymanskiW-sicomp76,DBLP:journals/jacm/Hunt82a}.

For showing \cref{undecidable-inter}, we rely on two ingredients. The
first is that infinity is decidable for higher-order pushdown
languages. This is a direct consequence of the decidability of the
more general simultaneous unboundedness
problem~\cite{DBLP:conf/icalp/Zetzsche15} or diagonal
problem~\cite{DBLP:journals/dmtcs/CzerwinskiMRZZ17}, which were shown
decidable for higher-order pushdown automata by Hague, Kochems and
Ong~\cite{HagueKochemsOng2016}.
\begin{lem}[\cite{HagueKochemsOng2016}]\label{finiteness-hopa}
  $\Inf{\calH}$ is decidable.
\end{lem}



The other ingredient is that turning binary representations into unary ones
can be achieved in higher-order pushdown automata.
\begin{lem}
  \label{lem:bin_inc_stack_depth}
  If $L\subseteq\{\ltr{0},\ltr{1}\}^*$ is an order-$k$ pushdown language, then $
  L'=\{\ltr{1}\ltr{0}^{\nu(w)} \; | \; w \in L \}$ is an order-$(k+2)$ pushdown
  language.
\end{lem}
\begin{proof}
  Let $\calA$ be an order-$k$ HOPA accepting $L\subseteq\{\ltr{0},\ltr{1}\}^*$.
  We construct an order-$(k+2)$ HOPA $\calA'$ for $L'$.
  We may clearly assume that $\calA$ has only one final state $q_f$.  The
  following diagram describes $\calA'$:
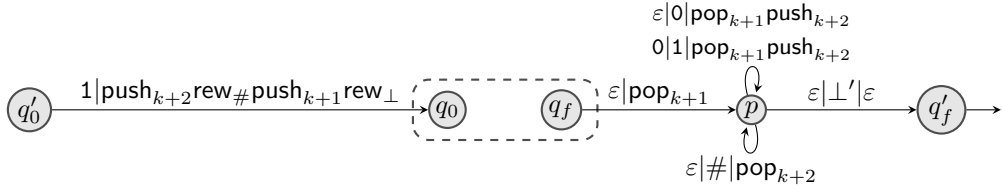
  \begin{center}
    \newhopa
  \end{center}
  The HOPA $\calA'$ starts in the configuration
  $(q'_0,\llbracket\bot'\rrbracket_{k+2})$ and in moving to $q_0$, it reads
  $\ltr{1}$ and goes to
  $(q_0,[\llbracket\bot'\rrbracket_{k+1}[\llbracket\#\rrbracket_k\llbracket \bot\rrbracket_k]_{k+1}]_{k+2})$.
  In the part in the dashed rectangle, $\calA'$ simulates $\calA$. However,
  instead of reading an input symbol $a\in\{\ltr{0},\ltr{1}\}$, $\calA$ stores
  that symbol on the stack. In order not to interfer with the simulation of
  $\calA$, this is done by copying the order-$k$ stack used by $\calA$ and
  storing $a$ in the copy below. This is achieved as follows. For every edge
  $p\xrightarrow{a|\gamma|v}q$ with $v\in(\instr{\Gamma}{k})^*$, $\calA'$
  instead has an edge
  \begin{center}
    \newedge
  \end{center}
  This pushes the input symbol $a$ on the (topmost order-$k$) stack, makes a
  copy of the topmost order-$k$ stack, removes the $a$ from this fresh copy, and
  then excutes $v$. Edges $p\xrightarrow{\varepsilon|\gamma|v}$ (i.e. ones that
  read $\varepsilon$ from the input) are kept.
  
  When $\calA'$ arrives in $q_f$, it has a stack
  $[\llbracket\bot'\rrbracket_{k+1}[\llbracket\#\rrbracket_k s_1\cdots
  s_{m}s]_{k+1}]_{k+2}$, where $s$ is the order-$k$ stack reached in the
  computation of $\calA$, and $s_1,\ldots, s_m$ store the input word
  $w\in\Sigma^*$ read by $\calA$, meaning
  $\tops(s_1)\cdots\tops(s_m)=w$. When moving to $p$, $\calA'$
  removes $s$ so as to obtain
  $[\llbracket\bot'\rrbracket_{k+1}[\llbracket\#\rrbracket_k s_1\cdots
  s_m]_{k+1}]_{k+2}$ as a stack.

  In $p$, $\calA'$ reads the input word $\ltr{0}^{\nu(w)}$ as follows. While in $p$,
  the stack always has the form
  \begin{equation} t=[\llbracket\bot'\rrbracket_{k+1} t_1\cdots
  t_\ell]_{k+2}, \label{powering-hopa-stacks}\end{equation}
  where each $t_i$ is an order-$(k+1)$ stack of the form
  $[\llbracket\#\rrbracket_k s_1\cdots s_m]_{k+1}$ for some order-$k$ stacks
  $s_1,\ldots,s_m\in\stacks{\Gamma}{k}$.  To formulate an invariant that holds
  in state $p$, we define a function $\mu$ on the stacks as in
  \eqref{powering-hopa-stacks}. First, if
  $t_i=[\llbracket\#\rrbracket_k s_1\cdots s_m]_{k+1}$, then let
  $\mu(t_i)=\nu(\tops(s_1)\cdots\tops(s_m))$. Next, let $\mu(t)=\mu(t_1)+\cdots+\mu(t_\ell)$.
  It is not hard to see that the loops on $p$ preserve the following invariant:
  If $\ltr{0}^r$ is the input word read from configuration $(p,t)$ to $(p,t')$, then
  $\mu(t)=r+\mu(t')$. To see this, consider a one step
  transition $(p,t) \xrightarrow{\varepsilon|\ltr{0}|\pop{k+1}
  \push{k+2}} (p,t')$. If
  \[t=[\llbracket\bot'\rrbracket_{k+1} t_1\cdots
  t_\ell]_{k+2}=[\llbracket\bot'\rrbracket_{k+1} t_1\cdots
  t_{\ell-1} [\llbracket\#\rrbracket_k s_1\cdots s_m]_{k+1}]_{k+2}\]
  then
   \[t'=[\llbracket\bot'\rrbracket_{k+1} t_1\cdots
  t_{\ell-1} [\llbracket\#\rrbracket_k s_1\cdots s_{m-1}]_{k+1}
  [\llbracket\#\rrbracket_k s_1\cdots s_{m-1}]_{k+1}]_{k+2}.\]
  If $w=\tops(s_1)\ldots\tops(s_m)$ then $w=w'\ltr{0}$ where $w'=
  \tops(s_1)\cdots\tops(s_{m-1})$ since we popped $s_m$ off the stack.  Moreover,
  \begin{flalign*}
    \mu(t')=\sum_{i=1}^{\ell-1}\mu(t_i) + 2\nu(w')=\sum_{i=1}^{\ell-1}\mu(t_i) + \nu(w)=\mu(t)
  \end{flalign*}
  Similarly we see that if the transition taken is $\ltr{0}|
  \ltr{1}|\pop{k+1}\push{k+2}$ then we get 
  $\mu(t')=\mu(t)+1$.
  By induction on the length of the run, we get $\mu(t)=r+\mu(t')$
  when $\ltr{0}^r$ is read.

  Now observe that when $\calA'$ first arrives in $p$ with stack $t$, then by
  construction we have $\ell=1$ and $\mu(t)=\mu(t_1)=\nu(w)$. Moreover, when
  $\calA'$ moves on to $q'_f$ with a stack as in \eqref{powering-hopa-stacks},
  then $\ell=0$ and thus $\mu(t)=0$. Thus, the invariant implies that if
  $\calA'$ reads $\ltr{0}^r$ while in $p$, then $r=\nu(w)$. This means, $\calA'$
  has read $\ltr{1}\ltr{0}^{\nu(w)}$ in total.

  Finally, from a stack $t$ as in
  \eqref{powering-hopa-stacks}, $\calA'$ reaches $q'_f$ in finitely
  many steps, please see \cref{proof:stack_termination}. 
 \end{proof}

\begin{lem}
	\label{lem:undec_inter}
	The problem $\Inter{\calH}{\cntlang{\power{\calH}}}$ is undecidable.
\end{lem}
\begin{proof}
  We reduce intersection emptiness for context-free languages,
  which is
  well-known to be undecidable \cite{hartmanis1967context}, to $\Inter{\calH}{\cntlang{\power{\calH}}}$. Let
  $K_1,K_2\subseteq\{\ltr{0},\ltr{1}\}^*$ be context-free. Since
  $K_1\cap K_2\ne\emptyset$ if and only if
  $\ltr{1}K_1\cap \ltr{1}K_2\ne\emptyset$ and $\ltr{1}K_i$ is context-free for
  $i=0,1$, we may assume that $K_1,K_2\subseteq
  \ltr{1}\{\ltr{0},\ltr{1}\}^*$. This implies $K_1\cap K_2\ne\emptyset$ if and
  only if $\nu(K_1)\cap\nu(K_2)\ne\emptyset$.

  Let $P_2=\N\setminus 2^{\N}~\cup~2^{\nu(K_2)}$. Then $P_2\subseteq\N$ is a
  power-$\calH$ predicate, because $\calH$ includes the context-free
  languages. Thus, the language $L_2=\{\ltr{1}\ltr{0}^{n}\mid n\in P_2\}$
  belongs to $\cntlang{\power{\calH}}$ and
  \[ L_2=\{\ltr{1}\ltr{0}^n\mid n\in \N\setminus 2^{\N}\}\cup \{\ltr{1}\ltr{0}^{2^{\nu(w)}} \mid w\in K_2\}.\]
  Moreover, let $L_1 := \{ \ltr{1}\ltr{0}^{2^{\nu(w)}} \; | \; w \in
  K_1\}$. Since
  $L_1=\{ \ltr{1}\ltr{0}^{\nu(\ltr{1}\ltr{0}^{\nu(w)})} \; | \; w \in K_1 \}$
  and $K_1$ is an order-$1$ pushdown language, applying
  \cref{lem:bin_inc_stack_depth} twice yields that $L_1$ is an order-$5$
  pushdown language and thus belongs to $\calH$. Now clearly
  $L_1\cap L_2\ne\emptyset$ if and only if $\nu(K_1)\cap\nu(K_2)\ne\emptyset$,
  which is equivalent to $K_1\cap K_2=\emptyset$.
\end{proof}

For showing decidability of regular separability, we use the following well-known
fact (please see \cref{appendix-regsep-unions} for a proof).
\begin{lem}\label{regsep-unions}
  Let $L=\bigcup_{i=1}^m L_i$ and $K=\bigcup_{i=1}^n K_i$. Then $\regsep{K}{L}$
  if and only if $\regsep{L_i}{K_j}$ for all $i\in\{1,\ldots,m\}$ and $j\in\{1,\ldots,n\}$.
\end{lem}
The last ingredient for our decision procedure is the following simple
but powerful observation from \cite{CzerwinskiZetzsche2019a} (for the
convenience of the reader, a proof can be found in \cref{appendix-move-transduction}).
\begin{lem}\label{move-transduction}
  Let $K\subseteq\Gamma^*$, $L\subseteq\Sigma^*$ and
  $T\subseteq\Sigma^*\times\Gamma^*$ be a rational transduction. Then
  $\regsep{L}{TK}$ if and only if $\regsep{T^{-1}L}{K}$.
\end{lem}
The following now completes the proof of \cref{undecidable-inter}.
\begin{lem}
  The problem $\RegSep{\calH}{\cntlang{\power{\calH}}}$ is decidable.
\end{lem}
\begin{proof}
  Suppose we are given $L_1\subseteq\Sigma^*$ from $\calH$ and
  $L_2\subseteq\Sigma^*$ from $\cntlang{\power{\calH}}$. Then we can write
  $L_2=\bigcup_{i=1}^n T_i\ltr{a}^{P_i}$, where for $1\le i\le n$,
  $T_i\subseteq \Sigma^*\times\ltr{a}^*$ is a rational transduction and
  $P_i\subseteq\N$ is a power-$\calH$ predicate. Since $\regsep{L_1}{L_2}$ if
  and only if $\regsep{L_1}{T_i\ltr{a}^{P_i}}$ for every $i$ (\cref{regsep-unions}), we may assume
  $L_2=T\ltr{a}^{P}$ for $T\subseteq\Sigma^*\times\ltr{a}^*$ rational and
  $P\subseteq\N$ a power-$\calH$ predicate. According to
  \cref{move-transduction}, $\regsep{L_1}{T\ltr{a}^P}$ if and only if
  $\regsep{T^{-1}L_1}{\ltr{a}^P}$. Since $T^{-1}$ is also a rational
  transduction and $\calH$ is a full trio, we may assume that $L_1$ is in
  $\calH$ with $L_1\subseteq \ltr{a}^*$ and $L_2=\ltr{a}^P$.

  By \cref{lem:unary_reg_sep_pow2}, we know that $\regsep{L_1}{\ltr{a}^P}$ if
  and only if $L_1$ is finite and disjoint from $\ltr{a}^P$. We can decide this
  as follows.  First, using \cref{finiteness-hopa} we check whether $L_1$ is
  finite. If it is not, then we know that $\regsep{L_1}{L_2}$ is not the case.
  
  If $L_1$ is finite, then we can compute a list of all words in $L_1$: We start
  with $F_0=\emptyset$ and then successively compute finite sets
  $F_i\subseteq L_1$. For each $i\in\N$, we check whether $L_1\subseteq F_i$,
  which is decidable because $L_1\cap (\ltr{a}^*\setminus F_i)$ is in
  $\calH$ and
  emptiness is decidable for $\calH$.  If $L_1\not\subseteq F_i$, then
  we
  enumerate words in $\ltr{a}^*$ until we find $\ltr{a}^m$ with $
  \ltr{a}^m\in L_1$ (membership in
  $L_1$ is decidable) and $\ltr{a}^m\notin F_i$. Then, we set
  $F_{i+1}=F_i\cup \{\ltr{a}^m\}$. Since $L_1$ is finite, this procedure must
  terminate with $F_i=L_1$.
  Now we have $\regsep{L_1}{\ltr{a}^P}$ if and only if
  $F_i\cap \ltr{a}^P=\emptyset$. The latter can be checked because
  $\power{\calH}$
  predicates are decidable.
\end{proof}



\section{Conclusion} 
\label{sec:conclusion}
We have presented a language class $ \mathcal{C}_1$ for which
intersection emptiness is decidable but regular separability is
undecidable in
\cref{sec:decidable_intersection_and_undecidable_regular_separability}.
Similarly, in
\cref{sec:decidable_regular_separability_and_undecidable_intersection}
we constructed $ \mathcal{C}_2,\mathcal{D}_2$ for which intersection
emptiness is undecidable but regular separability is decidable. All
three language classes enjoy good language theoretic properties in
that they are full trios and have a decidable emptiness
problem.
 
Let us provide some intuition on why these examples work. The
underlying observation is that intersection emptiness of two sets is
insensitive to the shape of their members: If $f\colon X\to Y$ is any
injective map and $S$ disjoint from the image of $f$, then for
$A,B\subseteq X$, we have $A\cap B=\emptyset$ if and only if
$(f(A)\cup S)\cap f(B)=\emptyset$. Regular separability, on the other
hand, is affected by such distortions: For example, if
$K,L\subseteq\ltr{1}\{\ltr{0},\ltr{1}\}^*$ are infinite, then
$\ltr{a}^{\N\setminus 2^{\N}}\cup \ltr{a}^{2^{\nu(K)}}$ and
$\ltr{a}^{2^{\nu(L)}}$ are never regularly separable, even if $K$ and
$L$ are.  Hence, roughly speaking, the examples work by distorting
languages (using encodings as numbers) so that intersection emptiness
is preserved, but regular separability reflects infinity of the input
languages. We apply this idea to language classes where intersection is
decidable, but infinity is not (\cref{undecidable-sep}) or the other way around
(\cref{undecidable-inter}). All this suggests that regular
separability and intersection emptiness are fundamentally different
problems.

Moreover, our results imply any simple combinatorial decision problem
that characterizes regular separability has to be incomparable with
intersection emptiness.  Consider for example the \emph{infinite
  intersection problem} as a candidate. It asks whether two given
languages have an infinite intersection.  Note that for
$L \in \calC, K \in \calD$ we have $L \cap K \neq \emptyset$ if and
only if $L\#^*$ and $K\#^*$ (where $\#$ is a symbol not present in $L$
or $K$) have infinite intersection. Moreover, as full trios, $\calC$
and $\calD$ effectively contain $L\#^*$ and $K\#^*$,
respectively. This implies a counterexample with decidable regular
separability and undecidable infinite intersection.

While the example from
\cref{sec:decidable_intersection_and_undecidable_regular_separability}
is symmetric (meaning: the two language classes are the same) and
natural, the example in
\cref{sec:decidable_regular_separability_and_undecidable_intersection}
is admittedly somewhat contrived: While pseudo-$\calC$ predicates rely
on the common conversion of binary into unary representations,
power-$\calC$ predicates are a bit artificial.  It would be
interesting if there were a simpler symmetric example with decidable
regular separability and undecidable intersection.


\bibliography{main}
\appendix
\section{Proof of Lemma \ref{thm:LCM_closure_rat_trans}}
\label{proof:reset_clo}
  Closure under union is easy to see, so let us show closure under
  rational transductions and intersection. Suppose
  $\calV=(Q,\Sigma,n,E,q_0,F)$ is a reset VASS and a
  transducer
  $\calT=(\tilde{Q},\Gamma,\Sigma,\tilde{E},\tilde{q}_0,\tilde{F})$
  accepting $T\subseteq\Gamma^*\times\Sigma^*$. We now think of
  $\calV$ producing output instead of reading input, which will be
  reflected by our formulations: Here, we say that an edge
  $(q,w,i,x,q)$ \emph{produces the word $w$}. This is because we
  imagine that $\calV$ is producing a word that is read and
  transformed by $\calT$, which then produces some output on its own.

  We may assume that every edge in $\calV$ is of the form
  $(q,a,i,x,q')$ with $a\in\Sigma\cup\{\varepsilon\}$. Moreover, we
  may assume that every edge in $\calT$ is of the form
  $(q,a,\varepsilon,q')$ with $a\in\Gamma$ or of the form
  $(q,\varepsilon,a,q')$ with $a\in\Sigma$.

  Our reset VASS $\calV'$ for $T(\langof{\calV})$ is obtained
  using a simple product construction. The state set of $\calV'$ is
  $Q'=Q\times \tilde{Q}$ and in every edge of $\calV'$, we either
  (i)~advance both $\calV$ and $\calT$---because $\calT$ is reading a
  symbol produced by $\calV$---or (ii)~only $\calT$---because $\calT$
  is producing some output---or (iii)~only $\calV$ because it makes a
  move that does not produce output.
%
  Formally, we include the following edges in $\calV'$:
  \begin{alignat*}{3}
    & \text{(i)}&\quad &((p,\tilde{p}),\varepsilon,i,x,(q,\tilde{q})) &\quad &\text{for each $(q,a,i,x,q)\in E$ and $(\tilde{p},\varepsilon,a,\tilde{q})\in \tilde{E}$ with $a\in\Sigma$,} \\
    & \text{(ii)}&&((p,\tilde{p}),a,1,0,(p,\tilde{q})) & & \text{for each $(\tilde{p},a,\varepsilon,\tilde{q})\in \tilde{E}$ and each $p\in Q$}\\
   & \text{(iii)}& &((p,\tilde{p}),\varepsilon,i,x,(q,\tilde{p}))& &\text{for each $(p,\varepsilon,i,x,q)\in E$ and each $\tilde{p}\in\tilde{Q}$.}  
 \end{alignat*}
 The state $(q_0,\tilde{q}_0)$ is the initial state of $\calV'$ and
 the states in $F\times\tilde{F}$ are final.  Then we clearly have
 $\langof{\calV'}=T(\langof{\calV})$.

 Let us now show that $\calR$ is closed under intersection and let
 $\calV_i=(Q_i,\Sigma,n,E,q_0^{(i)},F_i)$ for $i\in\{0,1\}$ be reset VASS (it is no loss of generality to assume that they both have
 $n$ counters). Again, we construct a reset VASS $\calV'=(Q',\Sigma,2n,E',q_0',F')$
 with $\langof{\calV'}=\langof{\calV_0}\cap\langof{\calV_1}$ using a
 simple product construction with $Q'=Q_0\times Q_1$. We may assume that every edge is
 $\calV_0$ and $\calV_1$ is of the form $(q,\varepsilon,i,x,q')$ or
 $(q,a,1,0,q')$, meaning that either an edge reads input or it
 operates on some counter. $\calV'$ has $2n$ counters and has three types of edges:
 (i)~$\calV_0$ and $\calV_1$ both advance and both read some letter $a\in\Sigma$, (ii)~just $\calV_0$ advances and (iii)~just $\calV_1$ advances. Formally, this means we have
 \begin{alignat*}{3}
 &\text{(i)}&\quad  &((p_0,p_1),a,1,0,(q_0,q_1)) &\quad& \text{for edges $(p_i,a,1,0,q_i)\in E_i$ in $i=0,1$} \\
 &\text{(ii)}& &((p_0,p_1),\varepsilon,i,x,(q_0,p_1))&&\text{for each
 $(p_0,\varepsilon,i,x,q_0)\in E_0$ and $p_1\in Q_1$} \\
  &\text{(iii)}&   &((p_0,p_1),\varepsilon,n+i,x,(p_0,q_1))&&\text{for each $(p_1,\varepsilon,i,x,q_1)\in E_1$ and $p_0\in Q_0$}.
\end{alignat*}
The initial state of $\calV'$ is $q'_0=(q_0^{(0)},q_0^{(1)})$ and the
final states of $\calV'$ are $F'=F_0\times F_1\subseteq Q'$. Then
clearly $\langof{\calV'}=\langof{\calV_0}\cap\langof{\calV_1}$.

\section{Details of Proof of \cref{lem:posExist_presArithLCM_truth_dec}}
\label{proof:truth_problem}
We show by structural induction on the formula that every formula in
$\logic$ defines a relation which belongs to $\pseudo{\calR}$ and a
reset VASS for L can be effectively computed.

\textbf{Atomic formulae: } $S(x)$ where $S \in \pseudo{\calR}$
	is by
	definition and one can construct automata which recognise the
	relations $x+y=z$ and $x=1$ (see \cite{blumensath2004finite}). \\
	\textbf{Induction: } Let $\phi(\bar{x},\bar{y},\bar{z}) =\phi_1(
	\bar{x},\bar{y}) \wedge \phi_2(\bar{x},\bar{z})$. By induction
	hypothesis, $\phi_1(\bar{x},\bar{y})$ and $\phi_2(\bar{x},\bar{z})$
	define $\pseudo{\calR}$ relations $L_1$ and $L_2$ respectively.

        An $m$-ary relation $P'\subseteq\N^m$ is said to be the
        \textit{cyclindrification} of a $k$-ary relation
        $P\subseteq\N^k$ if there exist indices $i_1,\ldots,i_k$ with
        $1 \leq i_1<i_2<\ldots<i_k \leq m$ such that
        $P'=\{ \bar{w}\in\N^m \mid (w_{i_1},w_{i_2},\ldots ,w_{i_k})
        \in P\}$.  If $P$ is a pseudo-$\calR$ relation then any
        cylindrification $P'$ of $P$ is also a pseudo-$\calR$
        relation: It is easy to construct a rational transduction
        $T\subseteq (\{\ltr{0},\ltr{1},\Box\}^k)^*\times
        (\{\ltr{0},\ltr{1},\Box\}^m)^*$ with $L_{P'}=TL_P$ and $\calR$
        is closed under rational
        transductions by~\cref{thm:LCM_closure_rat_trans}.

	Let
	$L_1',L_2'$ respectively be
	the cylindrification of $L_1,L_2$ w.r.t. $\bar{x},\bar{y},\bar{z}$.
	Then the language defined by $\phi$ is $L_1' \cap L_2'$, which belongs to $\calR$ according to \cref{thm:LCM_closure_rat_trans}.\\
	Similarly $\phi(\bar{x},\bar{y},\bar{z}) =\phi_1(
	\bar{x},\bar{y}) \wedge \phi_2(\bar{x},\bar{z})$ is the union of the
	appropriate cylindrifications of the languages corresponding to
	$\phi_1$ and $\phi_2$ and again \cref{thm:LCM_closure_rat_trans}
	applies. \\
	Let $\phi(\bar{x})=\exists y \; \phi'(y,\bar{x})$. Let $R$ (resp.
	$R'$) be the
	relation defined by $\phi$ (resp. $\phi'$). By
	induction hypothesis $L_R \in \calR$ and since $L_{R'}$ is a
	homomorphic image of $L_R$, \cref{thm:LCM_closure_rat_trans} it
	tells us that $L_{R'}$ belongs effectively to $\calR$.

\section{Details of Proof of \cref{lem:bin_inc_stack_depth}}
\label{proof:stack_termination}
We want to show that the machine $\calA'$ reaches $q'_f$ in finitely
  many steps. To accomplish this,
   we define a new parameter
  on the stack which is shown to strictly decrease on every
  transition. We start by defining the lexicographic order on finite
  sequence of numbers.  For $\bar{x}=(x_1,\ldots,x_n)\in\N^n$ and
  $\bar{y}=(y_1,\ldots,y_k)\in\N^k$, the lex order $\lelex$ is defined
  inductively by
  \[ \bar{x} \lelex \bar{y} \iff 
    x_1 < y_1 \text{ or } (x_1=y_1 \text{ and } (x_2,\ldots,x_n)
    \lelex (y_2,\ldots,y_k))\]
  Moreover, the empty sequence is smaller than any non-empty
  sequence.

  Using $\lelex$, we can now define an order on stacks. To this end, for an order-$(k+1)$ stack
  $s=[\llbracket \#\rrbracket_k,s_1\cdots s_m]_{k+1}$, we define its
  \emph{length} as $|s|=m+1$. Then, to an order-$(k+2)$ stack
  $t=[\llbracket\bot'\rrbracket t_1\cdots t_\ell]$, we associate the
  sequence $\sigma(t)=(|t_1|, \ldots, |t_\ell|)$. Finally, we set $t\lelex t'$
  if and only if $\sigma(t)\lelex\sigma(t')$.

  We can see that for any transition $(p,t) \xrightarrow{
  a|a'|v}(p,t')$ with $t=[\llbracket\bot'\rrbracket_{k+1}
  t_1\cdots
  t_{\ell-1} t_\ell]_{k+2}$ and $t_\ell=[\llbracket\#\rrbracket_k s_1\cdots s_m]_{k+1}$
  and $\sigma(t)=(x_1,x_2,\ldots,x_n)$, we have $t'\lelex t$:
  \begin{itemize}
    \item If $(a|a'|v)=(\varepsilon|\ltr{0}|\pop{k+1}\push{k+2})$
    then $\sigma(t')=(x_1,x_2,\ldots,x_{n-1},x_n-1,x_n-1)$.
    \item If $(a|a'|v)=(\ltr{0}|\ltr{1}|\pop{k+1}
    \push{k+2})$
    then $\sigma(t')=(x_1,x_2,\ldots,x_{n-1},x_n-1,x_n-1)$.
    \item If $(a|a'|v)=(\varepsilon|\#|\pop{k+2}
    )$
    then $\sigma(t')=(x_1,x_2,\ldots,x_{n-1})$.
  \end{itemize}
  Hence $\calA'$ must reach $q'_f$ in finitely many steps.

  \section{Proof of \cref{regsep-unions}}\label{appendix-regsep-unions}
  A separator witnessing $\regsep{X}{Y}$ also witnesses $\regsep{X_i}{Y_j}$
  for every $i\in[1,n]$ and $j\in[1,m]$. This shows the ``only if'' direction.

  For the ``if'' direction, suppose $R_{i,j}\subseteq M$ satisfies
  $X_i\subseteq R_{i,j}$ and $R_{i,j}\cap Y_j=\emptyset$. We claim that
  $R=\bigcup_{i=1}^n \bigcap_{j=1}^m R_{i,j}$ witnesses $\regsep{X}{Y}$.

  Since $X_i\subseteq R_{i,j}$ for every $i\in[1,n]$, we have
  $X_i\subseteq\bigcap_{j=1}^m R_{i,j}$ and hence
  $X=\bigcup_{i=1}^n X_i\subseteq R$. On the other hand, for every $i\in[1,n]$ and $k\in[1,m]$,
  we have $Y_k\cap R_{i,k}=\emptyset$ and thus $Y_k\cap \bigcap_{j=1}^m R_{i,j}=\emptyset$.
  This implies 
  \[ Y\cap R=\bigcup_{k=1}^m Y_k\cap \bigcup_{i=1}^n\bigcap_{j=1}^m R_{i,j}=\bigcup_{k=1}^m \bigcup_{i=1}^n \underbrace{Y_k\cap \bigcap_{j=1}^m R_{i,j}}_{=\emptyset}=\emptyset. \]

  \section{Proof of \cref{move-transduction}}\label{appendix-move-transduction}
  The following proof is from \cite{CzerwinskiZetzsche2019a}, but we include it for the convenience of the reader.
    \begin{proof}
      Suppose $L\subseteq R$ and $R\cap TK=\emptyset$ for some regular
      $R$. Then clearly $T^{-1}L\subseteq T^{-1}R$ and $T^{-1}R\cap
      K=\emptyset$. Therefore, the regular set $T^{-1}R$ witnesses
      $\regsep{T^{-1}L}{K}$. Conversely, if $\regsep{T^{-1}L}{K}$, then
      $\regsep{K}{T^{-1}L}$ and hence, by the first direction,
      $\regsep{(T^{-1})^{-1}K}{L}$. Since $(T^{-1})^{-1}=T$, this reads
      $\regsep{TK}{L}$ and thus $\regsep{L}{TK}$.
    \end{proof}

\end{document}